\documentclass[11pt]{article}
\usepackage[T1]{fontenc}
\usepackage{lmodern}
\usepackage{fullpage}
\usepackage{amsmath,amsthm,amssymb,color}
\usepackage{graphicx}
\usepackage{hyperref}
\hypersetup{colorlinks=true, citecolor=blue, linkcolor=red, urlcolor=blue, pdftitle={Optimal quantum adder}}

\newtheorem{theorem}{Theorem}
\newtheorem{lemma}[theorem]{Lemma}

\title{Ancilla-free optimal quantum adder}
\author{
Junhong Nie\thanks{Shandong University, Jinan 250100, China.}
\and
Jialiang Tang\thanks{State Key Laboratory of Processors, Institute of Computing Technology, Chinese Academy of Sciences, Beijing 100190, China; School of Computer Science and Technology, University of Chinese Academy of Sciences, Beijing 100049, China.}
\and
Xiaoming Sun\footnotemark[2]
\and
Wei Zi\thanks{Quantum Science Center of Guangdong-Hong Kong-Macao Greater Bay Area, Shenzhen, China. Corresponding author. Email: \nolinkurl{ziwei.quantum@outlook.com}.}
}

\date{}

\begin{document}
\maketitle

\begin{abstract}
We construct an exact $n$-bit quantum adder with $O(n)$ gates and $O(\log n)$ depth that requires no ancillary qubits. The circuit uses only $X$, CNOT, and Toffoli gates under all-to-all connectivity and is classically reversible. Our construction accelerates the controlled ladder underlying addition. Grouping the data into blocks and use existing control qubits as conditional workspace gives a recursive construction with $O(n\log\log n)$ gates and $O(\log n)$ depth. We then coordinate inter-block propagation via two short-path forest operations to cut the per-block workspace requirement, yielding a nonrecursive implementation with $O(n)$ gates. Both the gate count and circuit depth are asymptotically optimal.
\end{abstract}

\section{Introduction}
\label{sec:introduction}

An $n$-bit quantum adder implements
\begin{equation}
 \operatorname{ADD}_n:\quad
 (a,b)
 \longmapsto
 (a,a+b\bmod 2^n),
 \label{eq:adder}
\end{equation}
for $0\le a,b<2^n$, coherently on arbitrary superpositions. Quantum addition is a fundamental building block of quantum algorithms. It underlies modular arithmetic in algorithms for factoring and discrete logarithms~\cite{Shor1997,HaenerEtAl2020}, phase rotations in fault-tolerant quantum chemistry simulation~\cite{JonesEtAl2012}, and arithmetic oracles for block-encoding structured matrices~\cite{SuenderhaufEtAl2024}. These algorithms call addition many times, so the cost of each addition affects the gate count, depth, and ancillary space of the full algorithm. Reducing these costs together is therefore a basic problem in quantum circuit design.

Any exact $n$-bit adder over bounded-arity gates requires $\Omega(n)$ gates and $\Omega(\log n)$ depth, since the most significant output bit depends on all lower input bits. Our construction attains both bounds without ancillary qubits. The circuit uses only $X$, CNOT, and Toffoli gates and is therefore classically reversible.

\begin{theorem}
\label{thm:main}
With all-to-all connectivity, there is an exact $n$-bit quantum adder using $O(n)$ $X$, CNOT, and Toffoli gates and $O(\log n)$ depth, with no ancillary qubits.
\end{theorem}

\subsection{Related works}
\label{sec:related}

Ripple-carry adders pass carry information from one bit to the next, giving linear size and depth with little workspace. Cuccaro et al.~\cite{CuccaroEtAl2004} used one ancillary qubit, and Takahashi, Tani, and Kunihiro~\cite{TakahashiTaniKunihiro2010} achieved both bounds without ancillary qubits. Draper's Fourier-based adder~\cite{Draper2000} also avoids ancillary qubits and uses controlled phase rotations.

Carry-lookahead adders compute carry information in parallel, using additional workspace to reduce depth. Draper et al.~\cite{DraperEtAl2006} gave an adder with $O(n)$ gates and $O(\log n)$ depth using $O(n)$ ancillary qubits. Takahashi, Tani, and Kunihiro~\cite{TakahashiTaniKunihiro2010} reduced the ancillary space to $O(n/\log n)$ with the same size and depth bounds.

To reduce depth without ancillary qubits, Remaud and Vandaele~\cite{RemaudVandaele2025} accelerated the CNOT and Toffoli ladders in ripple-carry addition. Their adder uses $O(n\log n)$ gates and $O(\log^2 n)$ depth over classical reversible gates. Remaud~\cite{Remaud2025Structures} further related ripple-carry and carry-lookahead adders through Toffoli ladders.

\subsection{Main techniques}
\label{sec:techniques}

Our starting point is a standard reduction that performs addition using two controlled ladders and CNOT circuits of linear size and logarithmic depth~\cite{RemaudVandaele2025,Remaud2025Structures}. Section~\ref{sec:preliminaries} recalls this reduction. It leaves us with the task of accelerating the controlled ladder, which maps the data $x$ to $y$ while preserving its controls $p$:
\begin{equation}
 \mathcal T_n:(p,x)\longmapsto(p,y),\qquad
 y_0=x_0,\quad y_i=x_i+p_i y_{i-1}\quad(1\le i<n),
 \label{eq:ladder}
\end{equation}
where $p=(p_1,\ldots,p_{n-1})$ is the control register and all bit arithmetic is over $\mathbb F_2$: addition is XOR and multiplication is AND. In the direct implementation, each update uses the preceding output, so the gates run one after another. We reduce the depth by computing these contributions in parallel.

We separate propagation within blocks from propagation between blocks. For clarity, assume $n=(r+1)b$ and partition the data into $r+1$ consecutive blocks $B_j=(x_{j,1},\ldots,x_{j,b})$, $0\le j\le r$. The first index identifies the block and the second the position within it; $p_{j,t}$ controls the transition into that position, except that $p_{0,1}$ is absent. We use three steps: local ladders send $x$ to $z$; endpoint updates replace $z_{j,b}$ by $y_{j,b}$; and interior corrections replace the remaining $z_{j,t}$ by $y_{j,t}$. The first step is a smaller ladder problem. The third simply undoes the local ladder on the interior, adds the incoming contribution to its first bit, and runs the interior ladder again; Figure~\ref{fig:interior-correction} shows the circuit. We therefore focus on the second step. A contribution from the preceding endpoint reaches the end of block $B_j$ only when all controls in that block are one. Their product is the propagation control $\pi_j$. Write $\zeta_j=z_{j,b}$ and $\eta_j=y_{j,b}$ within this procedure. The endpoints themselves follow a ladder:
\begin{equation}
\begin{aligned}
 \eta_0&=\zeta_0,\qquad
 \eta_j=\zeta_j+\pi_j\eta_{j-1},\\
 \pi_j&=\prod_{t=1}^{b}p_{j,t}\quad(1\le j\le r).
\end{aligned}
 \label{eq:endpoint-ladder}
\end{equation}

No extra qubits are available for the endpoint update, so we borrow existing interior bits and restore them afterwards. These bits may start with arbitrary values. We use repeated additions to cancel their initial values while keeping the desired update; we call this an echo. To make this cancellation possible, the Pascal decomposition splits the inverse endpoint map $J$, which sends $\eta$ back to $\zeta$, into two operations that square to the identity:
\begin{equation}
 J=Q_1Q_0,\qquad Q_0^2=Q_1^2=I.
 \label{eq:pascal-overview}
\end{equation}
Four additions between the endpoints and a borrowed register apply each factor and restore that register. Applying $Q_1$ and then $Q_0$ gives the required endpoint map.

To compute these additions, we use one control wire in each block as a dirty flag, meaning that its initial value is arbitrary. Write $m_i$ for the original value of $p_{i,b}$. Complement this bit, then add the conjunction $\chi_i=\prod_{t=1}^{b-1}p_{i,t}$ into it. The resulting value $f_i=1+m_i+\chi_i$ satisfies $\pi_i=m_if_i$. To update an endpoint, we walk through the earlier endpoints and add their controlled contributions. The controls along any part of this walk form an interval product. The identity $\pi_i=m_if_i$ lets us compute it from products of the original and encoded flags.

The remaining controls provide conditional space~\cite{NieZiSun2024,KhattarGidney2025}. Every term added to $\zeta_i$ contains $\pi_i$, so the update vanishes when $\chi_i=0$. When $\chi_i=1$, complementing these controls gives zero bits for the calculation. Echoes using zero tests apply the calculation only when $\chi_i=1$, and two reads cancel the initial values of the borrowed data. A walk with $r$ stops needs $O(r)$ such bits. Taking $r,b=\Theta(\sqrt n)$ gives a linear-size, logarithmic-depth endpoint update. For these larger blocks, an incrementer correction keeps the interior step logarithmic in depth. Running the local ladders recursively then gives the following result.

\begin{theorem}
\label{thm:recursive}
With all-to-all connectivity, the controlled ladder and quantum addition admit exact ancilla-free circuits of size $O(n\log\log n)$ and depth $O(\log n)$ over $\{X,\mathrm{CNOT},\mathrm{Toffoli}\}$.
\end{theorem}

The extra size factor comes from doing linear work at each of $O(\log\log n)$ recursion levels. Each block needs $O(r)$ conditional space because its walk can visit all $r$ earlier endpoints, so the blocks must remain large.

To remove this factor, we decompose $J$ into two weighted forest maps of height $O(\log r)$. Their edges connect endpoints, with propagation products as weights. Each block now follows a walk of only $O(\log r)$ stops, so $O(\log r)$ conditional space per block suffices and we can take $b=\Theta(\log n)$. The local ladders and interior corrections then both run directly in logarithmic depth, eliminating recursion. A forest path is a ladder on possibly nonconsecutive endpoints. Each edge crosses a dyadic interval, whose length is a power of two and whose starting point is a multiple of that length. Its control is the product over that interval. We extend the one-block encoding to represent each such interval by one encoded bit and a constant number of original controls. The same Pascal factors, borrowed-register additions, and conditional-space circuit then implement the endpoint update with $O(n+r\log r)$ gates and $O(\log n)$ depth. This gives Theorem~\ref{thm:main}.

\section{Preliminaries}
\label{sec:preliminaries}

\subsection{Notation and elementary circuits}

We specify each reversible circuit by its action on tuples of bits. Its action on arbitrary quantum states then follows by linearity. We use single indices for the overall ladder and pairs $(j,t)$ within blocks: $j\in\{0,\ldots,r\}$ identifies the block, and $t\in\{1,\ldots,b\}$ the position within it. Registers are ordered by concatenating their blocks, and arithmetic registers start at the least significant bit. All logarithms are base two. In bit expressions, $+$ means XOR and multiplication means AND, so the arithmetic is over $\mathbb F_2$. The indicator $[H=0]$ is one exactly when every bit of $H$ is zero. Gate sequences are listed in execution order. Matrix products act on column vectors from right to left.

We count gates from $\{X,\mathrm{CNOT},\mathrm{Toffoli}\}$ and assume all-to-all connectivity. Gates in one layer must have disjoint supports, so no wire is used by two gates in the same layer. Borrowed bits provide temporary workspace on existing input wires. A dirty borrowed bit may initially have any value and may be entangled with other registers. The circuit must restore it on arbitrary quantum states, including these entangled inputs.

We handle long control products and increments with the following known circuits~\cite[Lemmas~1--3 and Theorem~4]{Vandaele2026}. A $k$-controlled $X$ and a $k$-bit incrementer each have size $O(k)$ and depth $O(\log(k+1))$, using one borrowed dirty bit. CNOT ladders, CNOT fanout (adding one source bit to several targets), and the common-control updates $t_i\leftarrow t_i+c u_i$ on disjoint pairs $(u_i,t_i)$ have linear size and logarithmic depth without additional wires. A negative control tests for zero. We obtain it by applying $X$ to the control wire before and after the controlled gate. An inverse circuit is obtained by reversing its gates, with the same size and depth.

The simplest implementation of $\mathcal T_n$ in \eqref{eq:ladder} runs the Toffoli updates $x_i\leftarrow x_i+p_i x_{i-1}$ in increasing order of $i$. It acts on $n$ data bits and $n-1$ controls, using $n-1$ gates and depth $n-1$. We use this direct implementation when the data register is small enough.

\subsection{Reduction to a controlled ladder}

To build the adder, we use two controlled ladders together with linear-size, logarithmic-depth CNOT circuits. We recall the reduction used in ripple-carry addition~\cite{TakahashiTaniKunihiro2010,RemaudVandaele2025}.

\begin{lemma}[Ladder-to-adder reduction]
\label{lem:reduction}
An implementation of $\mathcal T_n$ on its $2n-1$ input wires with size $S(n)$ and depth $D(n)$ gives an $n$-bit adder on its two input registers with size $2S(n)+O(n)$ and depth $2D(n)+O(\log(n+1))$.
\end{lemma}
\begin{proof}
The first ladder computes the carry information in the first register. The remaining stages use it to form the sum in the second register and restore the first. Define the difference map by $(\Delta A)_0=A_0$ and $(\Delta A)_i=A_i+A_{i-1}$ for $i\ge1$. Its inverse is the XOR prefix map. Let $\mathcal T_B$ denote the ladder on $A$ whose transition $i$ is controlled by $B_{i-1}$, and write $\bar B_i=1+B_i$. Starting from $(A,B)=(a,b)$, run the following seven stages:
\begin{equation}
\begin{array}{rll}
1.&B\leftarrow B+A;&2.\quad A\leftarrow\Delta A;\\
3.&A\leftarrow\mathcal T_B A;&4.\quad B\leftarrow B+A;\\
5.&A\leftarrow\mathcal T_{\bar B}^{-1}A;&6.\quad A\leftarrow\Delta^{-1}A;\\
7.&B\leftarrow B+A.&
\end{array}
\label{eq:adder-reduction}
\end{equation}
To see how the ladder computes the carry information, let $\rho_i=a_i+b_i$, and let $c_i$ be the carry entering position $i$ in integer addition, with $c_0=0$. The usual carry recurrence $c_{i+1}=a_i b_i+\rho_i c_i$ implies that $u_i=a_i+c_i$ satisfies
\begin{equation}
 u_0=a_0,\qquad
 u_i=a_i+a_{i-1}+\rho_{i-1}u_{i-1}\quad(1\le i<n).
 \label{eq:carry-ladder}
\end{equation}
Thus stage 3 computes $A=u$, and stage 4 XORs $u$ into $B$, leaving $B=\rho+u$. To restore $A$, we next erase its carry information. Stage 5 uses the complemented controls in $B$. The identity $(1+\rho_i+u_i)u_i=\rho_i u_i$ over $\mathbb F_2$ gives
\[
 A_0=a_0,\qquad
 A_i=u_i+(1+\rho_{i-1}+u_{i-1})u_{i-1}
     =a_i+a_{i-1}\quad(i\ge1).
\]
Stage 6 now restores $A=a$, and stage 7 leaves $B=\rho+c$, the sum bits. This uses two ladder calls. The difference and prefix maps have linear size and logarithmic depth. Complementing the controls takes two parallel $X$ layers, and the remaining updates are parallel CNOT layers. All stages use the two input registers, with no additional wires.
\end{proof}

\subsection{Pascal matrices}

The endpoint update must restore all borrowed bits. We will do this with echoes built from operations that square to the identity. The Pascal matrices give two such operations whose product is the inverse ladder when all controls are one. We first prove this uncontrolled identity, then insert the propagation controls in Section~\ref{sec:recursive-endpoints}.

\begin{lemma}[Pascal factorization]
\label{lem:pascal}
For $m\ge1$, let $L$ be the $m\times m$ matrix with ones on its diagonal and immediately below it, and zeros elsewhere. This is the inverse ladder with all controls equal to one. Then, modulo $2$,
\begin{equation}
 L=P_1P_0,
 \label{eq:pascal-factorization}
\end{equation}
where $P_0$ and $P_1$ are lower-triangular matrices given by
\begin{equation}
\begin{aligned}
 (P_0)_{ij}&=\binom{i}{j}\bmod2
     &&(0\le j\le i<m),\\
 (P_1)_{ij}&=\binom{i-1}{j-1}\bmod2
     &&(1\le j\le i<m).
\end{aligned}
\label{eq:pascal-matrices}
\end{equation}
Set $(P_1)_{00}=1$ and $(P_1)_{i0}=0$ for $i>0$. Moreover, $P_0^2=P_1^2=I$.
\end{lemma}
\begin{proof}
We first check that each factor squares to the identity. For $0\le j\le i<m$, compute the entries modulo $2$:
\[
\begin{aligned}
 (P_0^2)_{ij}
 &=\sum_{k=j}^{i}\binom{i}{k}\binom{k}{j}\\
 &=\binom{i}{j}\sum_{k=j}^{i}\binom{i-j}{k-j}
 =2^{i-j}\binom{i}{j}=[i=j].
\end{aligned}
\]
Thus $P_0^2=I$. The same calculation on coordinates $1,\ldots,m-1$ gives $P_1^2=I$. It remains to check that the product is $L$. For $i\ge1$, Pascal's identity gives
\[
\begin{aligned}
 (P_1P_0)_{ij}
 &\equiv\sum_{k=0}^{i-1}\binom{i-1}{k}\binom{k+1}{j}\\
 &=\sum_{k=0}^{i-1}\binom{i-1}{k}
   \left(\binom{k}{j}+\binom{k}{j-1}\right)\\
 &\equiv[i-1=j]+[i=j]\pmod2,
\end{aligned}
\]
where binomial coefficients outside their usual range are zero. Row zero is $(1,0,\ldots,0)$, so $P_1P_0=L$.
\end{proof}

\section{A block framework}
\label{sec:blocks}

We first run a ladder inside each block, ignoring the control from the preceding block. This computes the contributions from within each block. We then update the endpoints and use them to supply the missing contributions to the interiors. Figure~\ref{fig:block-framework} shows these three steps. Assume $n=(r+1)b$ and partition the data into $r+1$ consecutive blocks of size $b$:
\[
 B_j=(x_{j,1},\ldots,x_{j,b}),\qquad 0\le j\le r.
\]
For general $n$, take $r=\lceil n/b\rceil-1$ and adjust $B_0$ to $n-rb$ bits. The same construction and asymptotic bounds apply to this partition. Each column in the figure is a block. Its bottom wire $(j,b)$ is the endpoint; positions $1,\ldots,b-1$ form the interior. The colours track each data wire from its input $x_{j,t}$, through the local output $z_{j,t}$, to the final output $y_{j,t}$.

The control $p_{j,t}$ for $t\ge2$ connects positions $t-1$ and $t$ within $B_j$. For $j\ge1$, $p_{j,1}$ connects the preceding endpoint to the first position; $p_{0,1}$ is absent. Local ladders omit the incoming controls. The recursive construction uses $r=\Theta(\sqrt n)$, $b=\Theta(\sqrt n)$; the linear-size construction uses $r=\Theta(n/\log n)$, $b=\Theta(\log n)$.

\begin{figure}[!ht]
\centering
\includegraphics[scale=1.05]{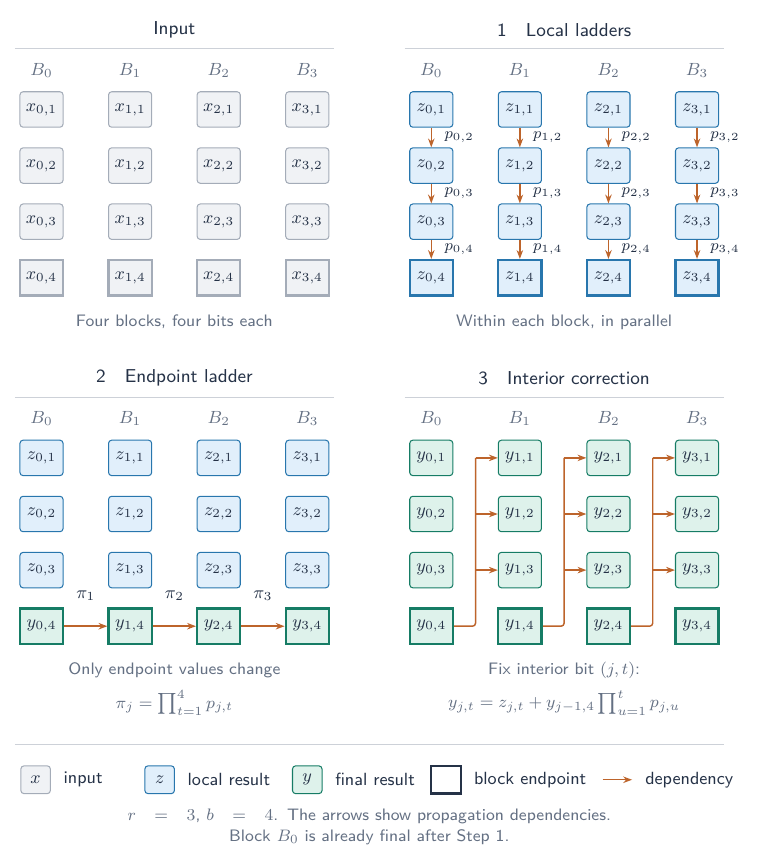}
\caption{The block framework with $r=3$ and $b=4$, so $n=16$. Each column represents one block, with its endpoint at the bottom. Arrows show propagation dependencies.}
\label{fig:block-framework}
\end{figure}

\begin{enumerate}
\item \textbf{Local ladders.} In panel 1, each column is processed independently: run $\mathcal T_b$ within $B_j$, using only $p_{j,2},\ldots,p_{j,b}$. The blue outputs are
\begin{equation}
 z_{j,1}=x_{j,1},\qquad
 z_{j,t}=x_{j,t}+p_{j,t}z_{j,t-1}
 \quad(2\le t\le b).
 \label{eq:block-local}
\end{equation}
The block circuits have disjoint supports and run in parallel. If a local ladder has size $S_{\mathrm{loc}}(b)$ and depth $D_{\mathrm{loc}}(b)$, the total size is $(r+1)S_{\mathrm{loc}}(b)$ and the depth is $D_{\mathrm{loc}}(b)$. We may use recursive local ladders or run their Toffoli gates directly in order. The direct implementation uses exactly $n-r-1$ Toffoli gates and has depth $b-1$.

\item \textbf{Endpoint update.} In panel 2, only the bottom entries change from blue to green. Write $\zeta_j=z_{j,b}$ and $\eta_j=y_{j,b}$ within this procedure. These endpoints form a ladder of their own:
\begin{equation}
 \eta_0=\zeta_0,\qquad
 \eta_j=\zeta_j+\pi_j\eta_{j-1},\qquad
 \pi_j=\prod_{t=1}^{b}p_{j,t}\quad(1\le j\le r).
 \label{eq:block-endpoints}
\end{equation}
Thus this step implements $(p,\zeta)\mapsto(p,\eta)$. The circuit evaluates each $\pi_j$ from the original block controls. It may temporarily modify those controls and borrow interior data, but it must restore both, leaving the interiors at $z_{j,t}$. Its size $S_{\mathrm{end}}$ and depth $D_{\mathrm{end}}$ include all control evaluation and restoration. The following sections give two implementations of this step.

\item \textbf{Interior correction.} In panel 3, each corrected endpoint supplies the missing contribution to the next block's interior. Block $B_0$ is already correct. For $j\ge1$, follow the three parts of Figure~\ref{fig:interior-correction}. First undo the local ladder on positions $1,\ldots,b-1$, recovering $x_{j,1},\ldots,x_{j,b-1}$. Next add $p_{j,1}y_{j-1,b}$ to the first bit with one Toffoli gate. Finally, run the interior ladder forward again. It carries this incoming contribution through the block, giving
\begin{equation}
 y_{j,t}=z_{j,t}
   +y_{j-1,b}\prod_{u=1}^{t}p_{j,u},
 \qquad 1\le t<b.
 \label{eq:block-interiors}
\end{equation}
\begin{figure}[!ht]
\centering
\includegraphics[scale=1.05]{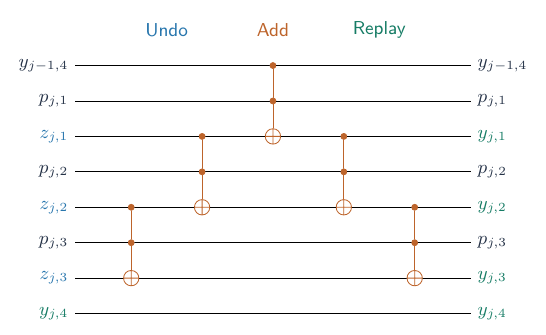}
\caption{Interior correction for $b=4$ and $j\ge1$. Undo the two interior gates, add the incoming contribution, and replay the two gates. Both endpoints and all controls are preserved.}
\label{fig:interior-correction}
\end{figure}

The circuit leaves both endpoints and all controls unchanged and needs no borrowed bits. Each block uses its incoming endpoint and leaves its outgoing endpoint to the next block, so all corrections run in parallel, giving total size $O(n)$ and depth $O(b)$.
\end{enumerate}

\begin{lemma}[Block reduction]
\label{lem:block-reduction}
For the partition above, suppose the local ladders have ancilla-free implementations of size $S_{\mathrm{loc}}(b)$ and depth $D_{\mathrm{loc}}(b)$. Suppose the endpoint update has size $S_{\mathrm{end}}$ and depth $D_{\mathrm{end}}$, using only existing wires and restoring the controls and all borrowed interior data. Then $\mathcal T_n$ has an ancilla-free implementation satisfying
\begin{equation}
\begin{aligned}
 S(n)&\le(r+1)S_{\mathrm{loc}}(b)
          +S_{\mathrm{end}}+O(n),\\
 D(n)&\le D_{\mathrm{loc}}(b)
          +D_{\mathrm{end}}+O(b).
\end{aligned}
\label{eq:block-costs}
\end{equation}
The incrementer correction in Section~\ref{sec:recursive-bounds} reduces the last depth term to $O(\log(b+1))$.
\end{lemma}

\section{The recursive construction}
\label{sec:recursive}

The block framework leaves us with the endpoint update in Step 2 of Section~\ref{sec:blocks}. We first reduce this update to additions between the endpoint register and a borrowed register. We then compute these additions in parallel inside the blocks, restoring all controls and borrowed interior data.

\subsection{Reducing the endpoint ladder to additions}
\label{sec:recursive-endpoints}

Recall that the endpoints must satisfy $\eta_i=\zeta_i+\pi_i\eta_{i-1}$, with $\eta_0=\zeta_0$. We work with the inverse map
\begin{equation}
 J:\eta\longmapsto\zeta,\qquad
 \zeta_0=\eta_0,\quad
 \zeta_i=\eta_i+\pi_i\eta_{i-1}\quad(1\le i\le r),
 \label{eq:inverse-endpoints}
\end{equation}
because its Pascal decomposition will let us restore the borrowed bits. Reversing its circuit gives the required update. For the controls between two endpoints, write
\[
 \Pi(j,i)=\prod_{\nu=j+1}^{i}\pi_\nu
 \quad(0\le j\le i\le r),\qquad \Pi(i,i)=1.
\]
The circuit evaluates these products from the original controls.

We have only borrowed interior bits for this update. We therefore decompose $J$ into operations that square to the identity, so that an echo can apply each operation and restore the borrowed bits. The Pascal factorization gives this decomposition.

\paragraph{Pascal decomposition.} Weight each entry of the matrices in Lemma~\ref{lem:pascal} by the control product between its indices. For $P_0$, this gives
\begin{equation}
 (Q_0)_{ij}=\left(\binom{i}{j}\bmod2\right)\Pi(j,i)
 \quad(0\le j\le i\le r),
 \label{eq:weighted-pascal}
\end{equation}
with zeros above the diagonal. Define $Q_1$ in the same way using $P_1$. In a matrix product, the weights from $j$ to $k$ and from $k$ to $i$ multiply to $\Pi(j,i)$, independent of $k$. The identities of Lemma~\ref{lem:pascal} therefore give
\begin{equation}
 J=Q_1Q_0,\qquad Q_0^2=Q_1^2=I.
 \label{eq:weighted-pascal-identities}
\end{equation}
Applying $Q_1$ and then $Q_0$ implements $J^{-1}$.

\paragraph{Using a borrowed register.} We focus on $Q_0$; the construction for $Q_1$ is the same with both Pascal indices decreased by one and coordinate zero fixed. Write $Q_0=I+A$. Our goal is to apply $I+A$ through additions between $\zeta$ and the borrowed $Z$, restoring $Z$ afterwards. We already have $A^2=(Q_0+I)^2=0$. For the cancellation, we also use the fixed backward shift
\[
 (Du)_i=u_{i+1}\quad(0\le i<r),\qquad (Du)_r=0.
\]
Thus, in $DA$, the output of row $i$ is added at position $i-1$. This shift makes the echo work through the identity $ADA=A$. For $i>j$,
\begin{equation}
\begin{aligned}
 (ADA)_{ij}
 &\equiv\Pi(j,i)\sum_{k=j}^{i-1}
       \binom{i}{k}\binom{k+1}{j}\\
 &\equiv\Pi(j,i)\binom{i}{j}=A_{ij}\pmod2.
\end{aligned}
\label{eq:chain-echo-identities}
\end{equation}
The two weights cover the full interval from $j$ to $i$ and both contain $\pi_{k+1}$. Since $\pi_{k+1}^2=\pi_{k+1}$, their product is $\Pi(j,i)$. For the second equality, add the terms $k=i$ and, when $j>0$, $k=j-1$. The enlarged sum is zero by Pascal's rule and the calculation in Lemma~\ref{lem:pascal}. Removing these terms leaves $\binom{i+1}{j}+\binom{i}{j-1}\equiv\binom{i}{j}\pmod2$. All other entries vanish. Decreasing the Pascal indices gives the same two identities for $A=Q_1-I$.

These identities give a circuit using $r+1$ borrowed interior bits for $Z$, disjoint from the endpoints:
\begin{equation}
\begin{aligned}
 (\zeta,Z)
 &\xrightarrow{\zeta\leftarrow\zeta+AZ}
       (\zeta+AZ,Z)\\
 &\xrightarrow{Z\leftarrow Z+DA\zeta}
       (\zeta+AZ,Z+DA\zeta)\\
 &\xrightarrow{\zeta\leftarrow\zeta+AZ}
       (\zeta+A\zeta,Z+DA\zeta)\\
 &\xrightarrow{Z\leftarrow Z+DA\zeta}
       ((I+A)\zeta,Z).
\end{aligned}
\label{eq:register-echo}
\end{equation}
The state expressions use the initial values of $\zeta,Z$. In the second addition, $A^2=0$ removes the dependence on the first change to $\zeta$. In the third, the two copies of $AZ$ cancel and $ADA=A$ leaves the desired $A\zeta$. The last addition restores $Z$ because $DA^2=0$. Thus four additions apply the factor and restore the borrowed register, whatever its initial state.

We have reduced the endpoint update to eight additions, four per factor. For $A=Q_0-I$, the first addition is
\begin{equation}
 \zeta_i\longmapsto\zeta_i+
 \sum_{j=0}^{i-1}\left(\binom{i}{j}\bmod2\right)\Pi(j,i)Z_j
 \quad(1\le i\le r),
 \label{eq:chain-addition}
\end{equation}
with coordinate zero unchanged. It is enough to construct this addition while preserving $Z$ and restoring the controls and temporary data. The other additions exchange the registers or redirect row $i$ to position $i-1$.

\subsection{Parallel walks along the chain}
\label{sec:chain-walks}

To compute the control products in \eqref{eq:chain-addition}, we separate each block's last control from its remaining controls. Write
\begin{equation}
 H_i=(p_{i,1},\ldots,p_{i,b-1}),\qquad
 \chi_i=\prod_{t=1}^{b-1}p_{i,t},\qquad m_i=p_{i,b}.
 \label{eq:local-conditional-space}
\end{equation}
The block control is $\pi_i=m_i\chi_i$. We want to express this product using the last control's value before encoding and its value afterwards. Use that wire as a dirty flag and apply
\begin{equation}
 \mathcal E_i:\quad p_{i,b}\longleftarrow p_{i,b}+1+\chi_i.
 \label{eq:block-encoding}
\end{equation}
Call the new value $f_i$. If $m_i=1$, then $f_i=\chi_i$; if $m_i=0$, then $\pi_i=0$. Thus, on every input,
\begin{equation}
 f_i=1+m_i+\chi_i,\qquad \pi_i=m_if_i.
 \label{eq:block-encoded-flag}
\end{equation}
Figure~\ref{fig:block-encoding} implements this encoding with one $X$ and one multi-controlled $X$. The latter borrows and restores an interior bit $w_i$. All blocks run in parallel, so $\mathcal E=\prod_i\mathcal E_i$ has size $O(rb)$ and depth $O(\log(b+1))$. It preserves $H_i$ and is its own inverse.

\begin{figure}[!ht]
\centering
\includegraphics[scale=1.05]{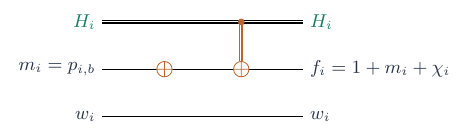}
\caption{Using $p_{i,b}$ as a dirty flag for $\chi_i$. The filled control on $H_i$ tests that all its bits are one. The multi-controlled $X$ borrows and restores $w_i$.}
\label{fig:block-encoding}
\end{figure}

To compute the sum for $\zeta_i$ in \eqref{eq:chain-addition}, follow the chain $i\to i-1\to\cdots\to0$. We call this a walk, with a stop at each earlier endpoint $j$. The contribution from $Z_j$ is controlled by
\begin{equation}
 \Pi(j,i)=M(j,i)F(j,i),\qquad
 M(j,i)=\prod_{\nu=j+1}^{i}m_\nu,\qquad
 F(j,i)=\prod_{\nu=j+1}^{i}f_\nu.
 \label{eq:chain-running-products}
\end{equation}
We therefore need to perform
\[
 \zeta_i\longmapsto\zeta_i+
 \sum_{j=0}^{i-1}\left(\binom{i}{j}\bmod2\right)
 M(j,i)F(j,i)Z_j.
\]
For example, the three Pascal coefficients before the diagonal in row $3$ are odd, so its update is
\[
 \zeta_3\longmapsto\zeta_3+
 M(2,3)F(2,3)Z_2+M(1,3)F(1,3)Z_1+M(0,3)F(0,3)Z_0.
\]
The same physical wires hold $m_\nu$ before $\mathcal E$ and $f_\nu$ afterwards. We therefore first compute $M$ and store it in borrowed data. This lets us keep $M$ while encoding the controls, then compute $F$ and combine the two products.

\paragraph{Computing $M$ before encoding.} To keep these products while $\mathcal E$ changes the controls, borrow a word $a_i=(a_{i,1},\ldots,a_{i,i})$ of interior data. Position $k=i-j$ corresponds to stop $j$. We will add $M(j,i)$ into $a_{i,k}$ and later cancel the initial value of this borrowed bit. For example, the three additions for the walk from $3$ are $m_3$, $m_3m_2$, and $m_3m_2m_1$.

We still need space to compute the products before adding them to $a_i$. Every term added to $\zeta_i$ contains $\pi_i=m_i\chi_i$, so the update is zero when $\chi_i=0$. When $\chi_i=1$, all bits of $H_i$ are one. Complementing them therefore gives $b-1$ zero bits for the calculation. We only need the calculation to take effect in this case. This is the conditional-space technique of Nie, Zi, and Sun and Khattar and Gidney~\cite{NieZiSun2024,KhattarGidney2025}. Figure~\ref{fig:chain-layout} locates $H_i$, the borrowed word $a_i$, and two borrowed bits $c_i,w_i$. The borrowed data are disjoint from $Z$, the endpoints, and the controls.

\begin{figure}[!ht]
\centering
\includegraphics[scale=1.05]{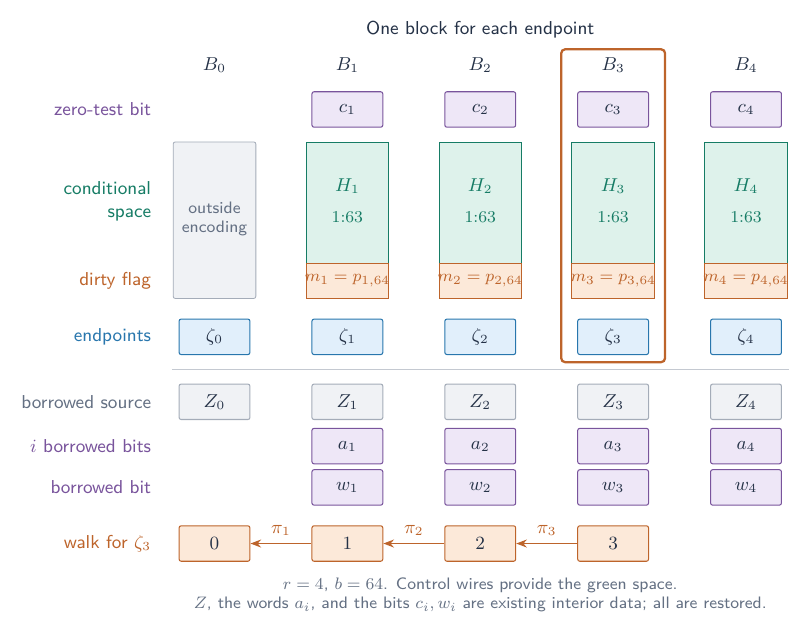}
\caption{The chain construction with $r=4$, $b=64$. The orange wire $p_{i,64}$ has original value $m_i$ and encoded value $f_i$. The green word $H_i$ contains positions $1,\ldots,63$. The borrowed bit $c_i$ above each block is used for zero tests. Below the endpoints are borrowed interior data: $Z_i$, the word $a_i$ of $i$ bits, and $w_i$. The highlighted walk for $\zeta_3$ visits $3,2,1,0$.}
\label{fig:chain-layout}
\end{figure}

When $H_i$ starts at zero, a circuit $U$ uses it to copy $m_i,m_{i-1},\ldots,m_1$ and compute all their running products, written $M(\cdot,i)=(M(i-1,i),\ldots,M(0,i))$. Figure~\ref{fig:prefix-scan} shows the calculation. First combine adjacent intervals up a balanced AND tree. Then work down the tree, passing to each interval the product of the inputs that precede it. At each leaf, multiply this preceding product by the leaf's own input. This gives every running product using $O(i)$ bits and gates and $O(\log(i+1))$ depth. We take $b\ge\kappa r$ for a sufficiently large fixed $\kappa\ge16$, so the copies, products, and sums below fit inside $H_i$, and the borrowed data fit inside the blocks. $U$ preserves all inputs outside $H_i$ and leaves its last position unused.

\begin{figure}[!ht]
\centering
\includegraphics[scale=1.05]{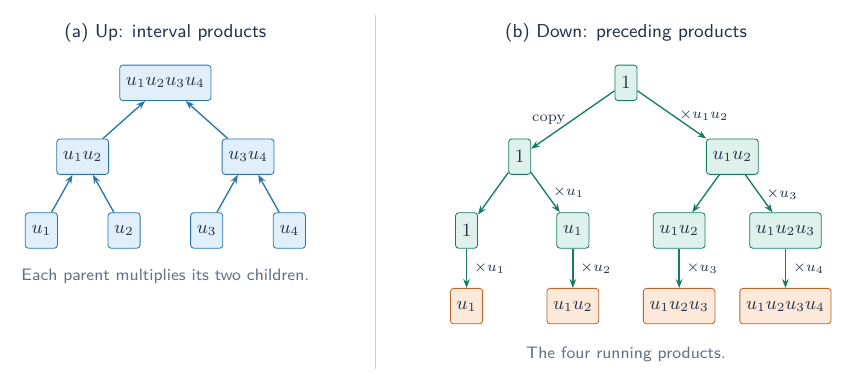}
\caption{Computing all running products of four copied inputs. The upward tree computes each interval's product. In the downward tree, a left child receives its parent's preceding product and a right child also multiplies the left interval's product. Multiplying at each leaf gives the four running products. All intermediate bits start at zero; the root preceding product is set to one by $X$. Each level takes a constant number of gate layers.}
\label{fig:prefix-scan}
\end{figure}

The calculation above assumes that $H_i$ starts at zero. The echo in Figure~\ref{fig:chain-m-echo} makes it take effect only when $\chi_i=1$, while restoring $H_i$ on every input. Complement $H_i$ and toggle $c_i$ by $[H_i=0]$. Run $U$, add $c_iM(\cdot,i)$ into $a_i$, and reverse $U$. Toggle $c_i$ by the same zero test and repeat these three operations, then undo the complements of $H_i$. If $\chi_i=1$, the zero test changes $c_i$ between the two calculations. Their contributions therefore add $M(\cdot,i)$ once, regardless of the initial value of $c_i$. If $\chi_i=0$, the zero test leaves $c_i$ unchanged. Both calculations then make the same addition, so they cancel even though $H_i$ began with nonzero bits. Thus the circuit performs
\[
 a_{i,k}\longmapsto a_{i,k}+\chi_iM(i-k,i)
 \quad(1\le k\le i),
\]
and restores $H_i,c_i,w_i$. Call the parallel application of these circuits to all blocks $\mathcal M$.

\begin{figure}[!ht]
\centering
\includegraphics[scale=1.05]{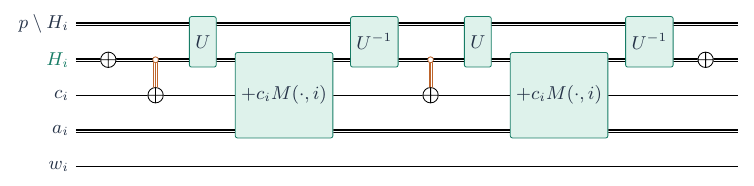}
\caption{Computing $M$ before encoding. The circuit for block $i$ adds $\chi_iM(\cdot,i)$ into the borrowed word $a_i$. $U$ computes the running products in $H_i$ when it starts at zero. An uncontrolled $\oplus$ complements its word, and each open control tests the entire word for zero. The zero tests borrow and restore $w_i$. All wires except $a_i$ return to their initial values.}
\label{fig:chain-m-echo}
\end{figure}

\paragraph{Computing $F$ after encoding.} The first stage has added $\chi_iM(\cdot,i)$ to the borrowed word $a_i$ and restored $H_i$. We can therefore apply $\mathcal E$ and read the encoded bits $f_\nu$. Compute their running products $F(j,i)$ with the same tree, then combine them with $a_i$ and $Z$ in the sum
\begin{equation}
 s_i=\sum_{j=0}^{i-1}
       \left(\binom{i}{j}\bmod2\right)a_{i,i-j}F(j,i)Z_j.
 \label{eq:chain-sum}
\end{equation}
When $H_i$ starts at zero, let $V$ copy the encoded bits, $a_i$, and $Z_{i-1},\ldots,Z_0$ into it, form the products, and compute this sum into its last position. Its products and sum take $O(i)$ gates and $O(\log(i+1))$ depth. It preserves every bit outside $H_i$.

The encoding leaves $H_i$ unchanged. The same echo can therefore still test whether its original bits were all one, the case $\chi_i=1$. Figure~\ref{fig:chain-f-echo} replaces $U$ by $V$ and each word addition by a Toffoli gate from $c_i,s_i$ to $\zeta_i$. When $\chi_i=1$, the two additions give $s_i$; when $\chi_i=0$, they cancel. Denote the parallel application of these circuits by $\mathcal F$. It performs $\zeta_i\mapsto\zeta_i+\chi_i s_i$, preserving $a_i,Z$ and restoring $H_i,c_i,w_i$. We can then apply $\mathcal E^{-1}$ to return the controls to their original values.

\begin{figure}[!ht]
\centering
\includegraphics[scale=1.05]{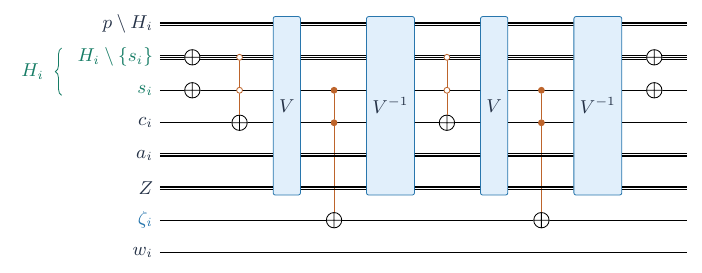}
\caption{Computing with $F$ after encoding. The circuit for block $i$ adds $\chi_i s_i$ to $\zeta_i$ using the same dirty echo. The bracketed word $H_i$ includes the position used for $s_i$. $V$ computes \eqref{eq:chain-sum} in that position when $H_i$ starts at zero. The borrowed data $a_i$ and $Z$ are preserved; all temporary calculations are reversed.}
\label{fig:chain-f-echo}
\end{figure}

\paragraph{Combining $M$ and $F$.} It remains to cancel the initial values in $a_i$. Figure~\ref{fig:path-echo} runs the four stages $\mathcal M,\mathcal E,\mathcal F,\mathcal E^{-1}$ twice. Suppose first that $\chi_i=1$. The first $\mathcal M$ changes $a_i$ to $a_i+M(\cdot,i)$, and the first $\mathcal F$ uses this value in its sum. The second $\mathcal M$ restores $a_i$, and the second $\mathcal F$ uses its original value. Thus the two sums cancel the unknown initial values in $a_i$. At stop $j$, with $k=i-j$, the contributions before applying the Pascal coefficient are
\begin{equation}
 (a_{i,k}+M(j,i))F(j,i)Z_j+a_{i,k}F(j,i)Z_j
 =\Pi(j,i)Z_j.
 \label{eq:two-read-cancellation}
\end{equation}
Here $a_{i,k}$ denotes its initial value. When $\chi_i=0$, both stages leave $a_i,\zeta_i$ unchanged, and every desired term is zero. This gives \eqref{eq:chain-addition} while preserving $Z$ and restoring all controls and borrowed bits.

\begin{figure}[!ht]
\centering
\includegraphics[scale=1.05]{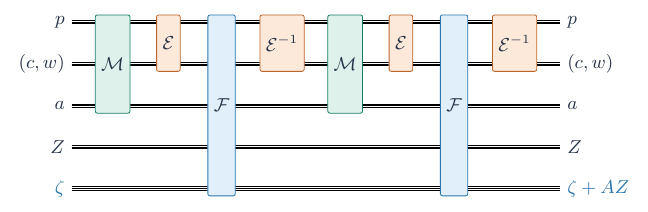}
\caption{The complete circuit for $(\zeta,Z)\mapsto(\zeta+AZ,Z)$. Read from left to right. $\mathcal M$ is the parallel application of Figure~\ref{fig:chain-m-echo}, and $\mathcal F$ that of Figure~\ref{fig:chain-f-echo}. Each $\mathcal M$ reads original controls; each $\mathcal F$ reads encoded controls. The two reads of $a$ cancel its initial values and retain the products $MF$. The words $a,c,w$ collect the borrowed bits of all blocks. This circuit will also be used for the forest walks.}
\label{fig:path-echo}
\end{figure}

\paragraph{Cost of the endpoint update.} Each of the $r$ walks has at most $r$ stops, giving $O(r^2)$ gates for copies, products, and sums. CNOT fanout trees copy shared inputs to the walks with $O(r^2)$ gates and $O(\log(r+1))$ depth. The additions controlled by the same $c_i$ use the common-control circuit from the preliminaries, with linear size and logarithmic depth. The zero tests on $H_i$ and the constant number of encoding calls use $O(rb)$ gates and $O(\log(b+1))$ depth. The other additions in \eqref{eq:register-echo} exchange the registers or redirect their outputs, and $Q_1$ changes the fixed Pascal selections.

\begin{lemma}[Recursive endpoint circuit]
\label{lem:recursive-endpoints}
For a sufficiently large fixed $\kappa\ge16$ and $b\ge\kappa r$, Step 2 has size $O(rb+r^2)$ and depth $O(\log(b+r+1))$. It uses only existing wires and restores all controls and borrowed interior data.
\end{lemma}

\subsection{Recursion and resource bounds}
\label{sec:recursive-bounds}

The recursive construction uses large blocks, so the direct interior correction would take too much depth. Here we use the incrementer circuit in Figure~\ref{fig:interior-echo} instead. Incrementing the little-endian word $(y_{j-1,b},p_{j,1},\ldots,p_{j,b})$ changes
\begin{equation}
 p_{j,t+1}\longmapsto p_{j,t+1}
   +y_{j-1,b}\prod_{u=1}^{t}p_{j,u},
 \qquad 1\le t<b.
 \label{eq:increment-prefix}
\end{equation}
The two CNOT layers add the old and new values of each control to the corresponding interior bit. Their original values cancel, leaving exactly the correction in \eqref{eq:block-interiors}. The inverse increment restores the incoming endpoint and all controls. Each incrementer borrows and restores the first interior bit. As before, different blocks use disjoint wires, so the corrections run in parallel with total size $O(n)$ and depth $O(\log(b+1))$.

\begin{figure}[!ht]
\centering
\includegraphics[scale=1.05]{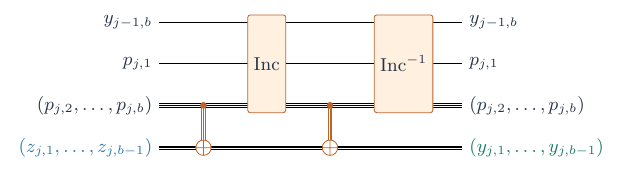}
\caption{The incrementer correction used for the larger recursive blocks. The two bundled wires each contain $b-1$ bits, and the CNOTs act bitwise. The increment and its inverse act on the little-endian word $(y_{j-1,b},p_{j,1},\ldots,p_{j,b})$, each borrowing and restoring the first interior bit.}
\label{fig:interior-echo}
\end{figure}

Each walk needs space proportional to its number of stops. We therefore choose blocks large enough to hold these calculations. Fix $\kappa$ as in Lemma~\ref{lem:recursive-endpoints}, and take $r\ge2$ to be the largest power of two satisfying $\kappa(r+1)r\le n$. Set
\[
 b=\left\lceil\frac{n}{r+1}\right\rceil.
\]
For the finitely many smaller sizes, use the direct ladder. Then $b\ge\kappa r$, $r=\Theta(\sqrt n)$, and $b=\Theta(\sqrt n)$. Adjust $B_0$ for the remainder as in Section~\ref{sec:blocks}. Its length $n-rb$ lies between $b-r$ and $b$, leaving enough interior data for its borrowed source bit.

The endpoint update now has size $O(n)$ and depth $O(\log n)$. Run the local ladders recursively in parallel and use the interior correction above. The resulting bounds are
\begin{equation}
\begin{aligned}
 S(n)&\le(r+1)S(b)+O(n),\\
 D(n)&\le D(b)+O(\log(n+1)).
\end{aligned}
 \label{eq:recursive-bounds}
\end{equation}
At each level, the blocks partition the data, so the work outside recursive calls is $O(n)$. At each level, the block size falls to its square root, up to a constant factor. It reaches constant size after $O(\log\log n)$ levels, giving size $O(n\log\log n)$. Along a branch, the logarithms of the block sizes decrease geometrically, so the depths sum to $O(\log n)$. All borrowed wires are restored within each call. Lemma~\ref{lem:reduction} gives the same bounds for addition, proving Theorem~\ref{thm:recursive}.

The remaining size factor comes from repeating linear work at every level. We used large blocks because each endpoint walked along a chain of up to $r$ stops. Shortening the walks to $O(\log r)$ stops will let us use smaller blocks. We can then keep the circuit of Figure~\ref{fig:path-echo} and complete the construction without recursion.

\section{The linear-size construction}
\label{sec:linear}

The recursive construction needs large blocks because each walk can have up to $r$ stops. We now shorten the walks to $O(\log r)$ stops. The same circuit then needs only $O(\log r)$ conditional space per block, so logarithmic block size suffices. We can run the local ladders directly and remove the recursion.

We keep the three steps of Section~\ref{sec:blocks}. For the endpoint update, decompose its inverse $J$ into two forests, apply the Pascal identities along their paths, and reuse the borrowed-register additions from \eqref{eq:register-echo}. To use these shorter walks, we need to encode the longer intervals crossed by forest edges. We also need to combine contributions from walks that reach the same endpoint.

\subsection{From one ladder to two forests}
\label{sec:forest-decomposition}

We replace the chain of consecutive endpoints by two forests whose paths are short. Their maps will satisfy $J=F_{\mathrm{in}}F_{\mathrm{out}}$. For $i>0$, let $v_2(i)$ be the largest integer $k$ such that $2^k$ divides $i$, and write
\[
 \lambda(i)=2^{v_2(i)},\qquad
 \Pi(a,b)=\prod_{\nu=a+1}^{b}\pi_\nu\quad(0\le a<b\le r).
\]
Thus $\lambda(i)$ is the value of the lowest set bit in the binary expansion of $i$. It determines the interval crossed by an edge. Both maps fix coordinate zero; for $i>0$, define
\begin{equation}
\begin{aligned}
 (F_{\mathrm{in}}u)_i
   &=u_i+\sum_{\substack{j>0\\j+\lambda(j)=i}}\Pi(j,i)u_j,\\
 (F_{\mathrm{out}}u)_i
   &=u_i+\Pi(i-\lambda(i),i)u_{i-\lambda(i)}.
\end{aligned}
\label{eq:forest-factorization}
\end{equation}
To check that their product is $J$, fix $i>0$, set $t=v_2(i)$, and write $j_k=i-2^k$ for $0\le k\le t$. The indices in the sum for $F_{\mathrm{in}}$ are $j_0,\ldots,j_{t-1}$. For $k<t$, we have $\lambda(j_k)=2^k$, so
\[
 \Pi(j_k,i)(F_{\mathrm{out}}u)_{j_k}
 =\Pi(j_k,i)u_{j_k}+\Pi(j_{k+1},i)u_{j_{k+1}}.
\]
The second term of each pair is the first term of the next pair, so these terms cancel modulo $2$. The last term cancels the extra term $\Pi(j_t,i)u_{j_t}$ in $(F_{\mathrm{out}}u)_i$, leaving $u_i+\Pi(i-1,i)u_{i-1}=u_i+\pi_i u_{i-1}$, as required by \eqref{eq:inverse-endpoints}. When $i$ is odd, $t=0$ and the sum is empty; $F_{\mathrm{out}}$ already gives this update.

Figure~\ref{fig:forest-decomposition} shows the two maps as forests: $F_{\mathrm{out}}$ has an edge $i-\lambda(i)\to i$, and $F_{\mathrm{in}}$ has an edge $j\to j+\lambda(j)$ whenever the successor is at most $r$. Each edge $a\to b$ carries weight $\Pi(a,b)$. In the left forest, each nonzero vertex has one parent; in the right, each vertex has at most one successor. Following parents on the left clears one set bit at each step. Following successors on the right moves the lowest set bit to a higher position at each step. Both walks therefore take at most a logarithmic number of steps, and the forests have height at most $h=\lceil\log_2(r+2)\rceil$.

To apply the Pascal identities along a path, we need consecutive indices for its vertices. We call these indices grades. In $F_{\mathrm{out}}$, give vertex $i$ grade $d_i$ equal to its distance from the root. In $F_{\mathrm{in}}$, use $h$ minus its distance to the sink, the vertex with no successor. Every directed edge then increases the grade by one. The green row labels in Figure~\ref{fig:forest-decomposition} show these grades. For example, the highlighted path $0\to4\to6\to7$ on the left has grades $0,1,2,3$, while $1\to2\to4\to8$ on the right has grades $1,2,3,4$. Each path is therefore a ladder indexed by consecutive grades, with vertices drawn from possibly nonconsecutive endpoints. Its step from $a$ to $b$ is controlled by $\Pi(a,b)$, the conjunction of all controls in blocks $B_{a+1},\ldots,B_b$. Successive edge intervals are disjoint and concatenate, so the weight of a path is
\[
 \prod_{\ell=1}^{m}\Pi(v_{\ell-1},v_\ell)=\Pi(v_0,v_m)
 \quad\text{for }v_0\to v_1\to\cdots\to v_m.
\]
The edge labels and products beneath the orange paths in Figure~\ref{fig:forest-decomposition} illustrate this correspondence.
\begin{figure}[!ht]
\centering
\includegraphics[scale=1.05]{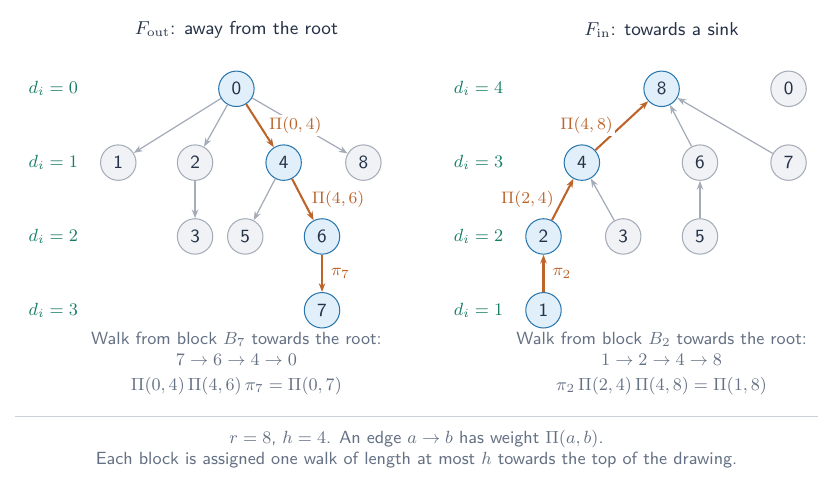}
\caption{The two forests on endpoints $0,\ldots,8$, with $h=4$. Green labels give the grade $d_i$ of every vertex in the corresponding row. The orange paths show the edges visited by one block's walk in each forest. Their edge weights multiply to the conjunction over the full path interval. Each path is a ladder on its ordered vertices, which may be nonconsecutive endpoints.}
\label{fig:forest-decomposition}
\end{figure}

We can now reuse Lemma~\ref{lem:pascal} on each path. Fix either forest $F$, and define its factor $Q_0$ by an identity diagonal and
\begin{equation}
 (Q_0)_{ij}=\left(\binom{d_i}{d_j}\bmod2\right)\Pi(j,i)
 \quad\text{if there is a nonempty directed path }j\to i.
\label{eq:forest-pascal}
\end{equation}
All other entries below the diagonal are zero. Define $Q_1$ in the same way with both grades decreased by one, setting its entries below the diagonal to zero when $d_j=0$. There is at most one directed path between any two vertices. Matrix multiplication along this path is exactly the weighted Pascal calculation from Section~\ref{sec:recursive-endpoints}, so
\begin{equation}
 F=Q_1Q_0,\qquad Q_0^2=Q_1^2=I.
\label{eq:forest-involutions}
\end{equation}
The parity of each binomial coefficient determines whether its path contribution is added. To apply $J^{-1}$, apply the two factors for $F_{\mathrm{in}}^{-1}$ and then those for $F_{\mathrm{out}}^{-1}$, using $Q_1$ followed by $Q_0$ in each case.

\subsection{Reusing the borrowed-register additions}
\label{sec:forest-dirty-register}

For a forest factor $Q=I+A$, we use the four additions in \eqref{eq:register-echo} to apply $Q$ and restore the borrowed $Z$. The identity $A^2=0$ follows again from $Q^2=I$. It remains to check $ADA=A$ for the same backward shift $D$ used in Section~\ref{sec:recursive-endpoints}.

For $Q_0$, expand $(ADA)_{ij}=\sum_{k=0}^{r-1}A_{ik}A_{k+1,j}$. In $F_{\mathrm{out}}$, a nonempty directed path can start only at an even vertex. The contributing indices $k$ lie on the path from $j$ to $i$, excluding $i$, and vertex $k+1$ is a child of $k$. In $F_{\mathrm{in}}$, a nonempty directed path can end only at a positive even vertex. Here the contributing vertices $k+1$ lie on the path after $j$, including $i$, and $k$ is the odd predecessor of $k+1$. In both cases $d_{k+1}=d_k+1$, and the possible grades $d_k$ range from $d_j$ to $d_i-1$. The weights multiply to $\Pi(k,i)\Pi(j,k+1)=\Pi(j,i)$, since the repeated factor $\pi_{k+1}$ satisfies $\pi_{k+1}^2=\pi_{k+1}$. Thus, on a nonempty path $j\to i$,
\begin{equation}
\begin{aligned}
 (ADA)_{ij}
 &\equiv\Pi(j,i)\sum_{d=d_j}^{d_i-1}
       \binom{d_i}{d}\binom{d+1}{d_j}\\
 &\equiv\Pi(j,i)\binom{d_i}{d_j}=A_{ij}\pmod2,
\end{aligned}
\label{eq:forest-echo-identities}
\end{equation}
where the second equality is exactly the chain identity \eqref{eq:chain-echo-identities}, with grades in place of indices. If there is no nonempty path, both entries are zero. For $Q_1$, decrease both grades by one. The columns whose source has grade zero are zero on both sides. The same four-addition circuit therefore applies every forest factor.

The two forests contribute four factors. Sixteen additions suffice for the endpoint ladder, all reusing the same restored $Z$. As before, we describe $\zeta\leftarrow\zeta+AZ$; the other additions exchange the registers or redirect row $i$ to position $i-1$.

\subsection{Parallel walks from the blocks}
\label{sec:shared-interval-encoding}
\label{sec:batched-paths}

The four-addition echo leaves us with the same task as before: implement
\[
 (\zeta,Z)\longmapsto(\zeta+AZ,Z),
\]
while restoring all temporary data. Here $\zeta$ contains the block endpoints and $Z$ is a borrowed register of $r+1$ interior bits. We first take $A=Q_0-I$ for the outward forest. Expanding the update gives, for each endpoint,
\[
 \zeta_i\longmapsto\zeta_i+
 \sum_{j:\,j\to i}
 \left(\binom{d_i}{d_j}\bmod2\right)\Pi(j,i)Z_j,
 \qquad 0\le i\le r,
\]
where the sum runs over nonempty directed paths $j\to i$. In the left forest of Figure~\ref{fig:forest-decomposition}, these sources $j$ are the ancestors of $i$ other than $i$ itself. Each term uses the source bit $Z_j$, controlled by the product of edge weights along the path from $j$ to $i$, namely $\Pi(j,i)$. The Pascal coefficient depends only on the grades, so the circuit knows in advance which terms to include.

For example, when $r=8$, the highlighted path in the left forest of Figure~\ref{fig:forest-decomposition} contains vertices $7,6,4,0$, with grades $3,2,1,0$, respectively. Their endpoint updates are
\begin{equation}
\begin{aligned}
 \zeta_7&\longmapsto\zeta_7+
 \pi_7Z_6+\Pi(4,7)Z_4+\Pi(0,7)Z_0,\\
 \zeta_6&\longmapsto\zeta_6+\Pi(0,6)Z_0,\\
 \zeta_4&\longmapsto\zeta_4+\Pi(0,4)Z_0,\\
 \zeta_0&\longmapsto\zeta_0.
\end{aligned}
\label{eq:walk-seven-goal}
\end{equation}
For $\zeta_7$, walk along $7\to6\to4\to0$. The edge controls are $\pi_7$, $\Pi(4,6)$, and $\Pi(0,4)$. At each stop, multiply the edge controls used so far. Use this product to control the addition of the source bit at that stop. All three Pascal coefficients are one. For $\zeta_6$, walk along $6\to4\to0$. Skip the addition at $4$ because $\binom21\equiv0\pmod2$, but retain the edge control $\Pi(4,6)$: at $0$, add $\Pi(4,6)\Pi(0,4)Z_0=\Pi(0,6)Z_0$. The walk $4\to0$ adds $\Pi(0,4)Z_0$ to $\zeta_4$, and the root $\zeta_0$ is unchanged. Each block $B_i$ has its own walk from vertex $i$ to the root, and we perform all walks in parallel.

These are the same running-product calculations as in Section~\ref{sec:chain-walks}, now with at most $h$ stops. A forest edge can cross many blocks. Computing its weight by scanning those blocks separately for each walk would undo the saving from shorter paths. All edge intervals belong to
\[
 \mathcal I=\{(j,i]\subseteq(0,r]:i-j=2^k,
                  \ j\equiv0\pmod{2^k},\ k\ge0\}.
\]
These are dyadic intervals: each has power-of-two length and starts at a multiple of that length. We encode them once and share their encoded bits among the walks. The construction of Nie, Zi, and Sun~\cite{NieZiSun2024} extends the one-block identity \eqref{eq:block-encoded-flag}. It represents each $\Pi(j,i)$ as the conjunction of one encoded bit and at most four selected original control bits. We call these selected controls its delegates.

We use conditional space inside each block for the same reason as in Section~\ref{sec:chain-walks}: every term added to $\zeta_i$ contains $\pi_i=\prod_{t=1}^{b}p_{i,t}$. For example, with $r=8$,
\[
 \Delta\zeta_7=\pi_7\bigl(
 Z_6+\Pi(4,6)Z_4+\Pi(4,6)\Pi(0,4)Z_0\bigr).
\]
If any control chosen for conditional space in $B_7$ is zero, the whole change vanishes. If all chosen controls are one, complementing them provides zero space. A walk now has only $h$ stops, so we need only $O(h)$ such bits in each block. For the outward walks, the space in $B_i$ is used only to compute $\zeta_i$.

Figure~\ref{fig:path-layout} shows the bits used by the circuit. Take $H_i=(p_{i,1},\ldots,p_{i,6h+2})$, assuming $b-4\ge6h+2$. The four control positions at the bottom of each block are for delegates. The first two can be chosen at a larger interval and passed to smaller intervals. The third and fourth are chosen at the one-block interval. Their physical indices are $b-3,b-2,b-1,b$. The words $a_i$, with one bit per stop, and the bits $c_i,w_i$ are borrowed from interior data, disjoint from $Z$ and the endpoints. They may start in arbitrary states and will be restored. The bit $c_i$ is drawn above $B_i$ because it controls the use of $H_i$.

\begin{figure}[!ht]
\centering
\makebox[\linewidth][c]{%
\includegraphics[scale=1.05]{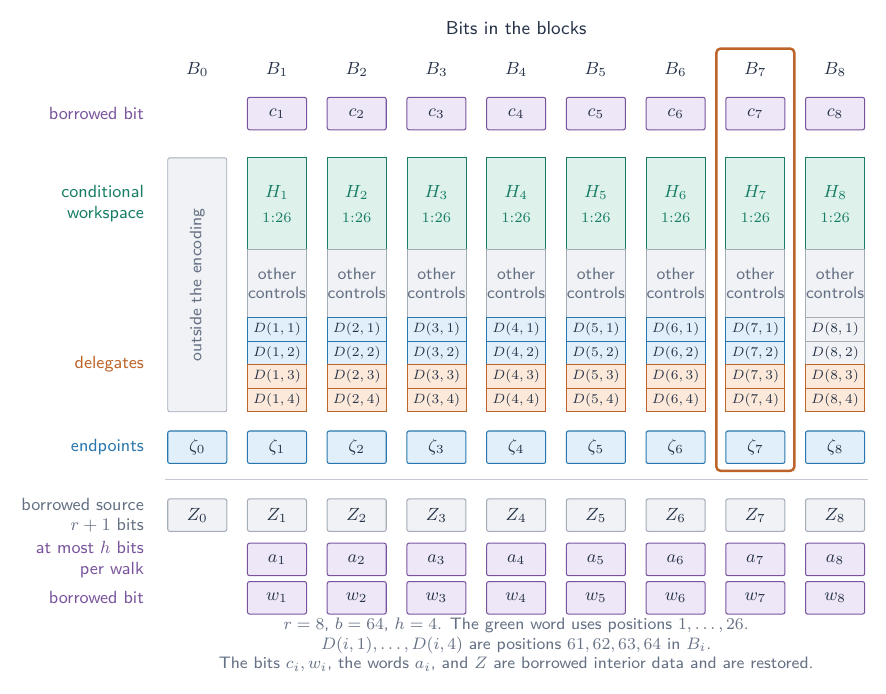}}
\caption{The layout for $r=8$, $b=64$, and $h=4$. Green bits form $H_i$. The four bottom positions are $D(i,1),\ldots,D(i,4)$, with physical coordinates $(i,61),\ldots,(i,64)$. Blue positions are chosen at larger intervals; orange positions are chosen at $B_i$. The grey positions in $B_8$ are left unchanged. The words $a_i$, the bits $c_i,w_i$, and $Z$ are existing interior data, all restored by the circuit.}
\label{fig:path-layout}
\end{figure}

\paragraph{Encoding the interval products.} Let $D(i,1),\ldots,D(i,4)$ name four distinct control positions in $B_i$. We first describe the recursion using these positions, then fix their coordinates. Each recursive call works on one interval and chooses two control bits there. Smaller calls omit positions already chosen at larger intervals. At an internal interval, the circuit tests the four bits chosen by its two children, complements them, runs the child circuits, and combines their encoded bits. This leaves the interval product in one encoded bit whenever its original delegates are all one.

Figure~\ref{fig:interval-encoding} shows two recursion layers on eight blocks. The root uses $p_{D(1,1)},p_{D(1,2)}$, and its children use the pairs in $B_3,B_5$. To see how the circuit works, consider the case where the root's two bits enter at zero. First compute
\[
 p_{D(1,2)}\longleftarrow p_{D(1,2)}+
 p_{D(3,1)}p_{D(3,2)}p_{D(5,1)}p_{D(5,2)}.
\]
Complement those four child bits, run the two child circuits, and then compute
\[
 p_{D(1,1)}\longleftarrow p_{D(1,1)}+
 p_{D(1,2)}p_{D(3,1)}p_{D(5,1)}.
\]
The next layer repeats these gates: the call on $(0,4]$ uses the pairs in $B_2,B_4$, and the call on $(4,8]$ uses those in $B_6,B_7$. At a one-block interval $B_i$, use $D(i,3),D(i,4)$ and add the conjunction of the controls outside all chosen positions into $p_{D(i,3)}$. Before starting the root, complement its two chosen bits. Every other pair is complemented by its parent before its call begins.

\begin{figure}[!ht]
\centering
\includegraphics[scale=1.05]{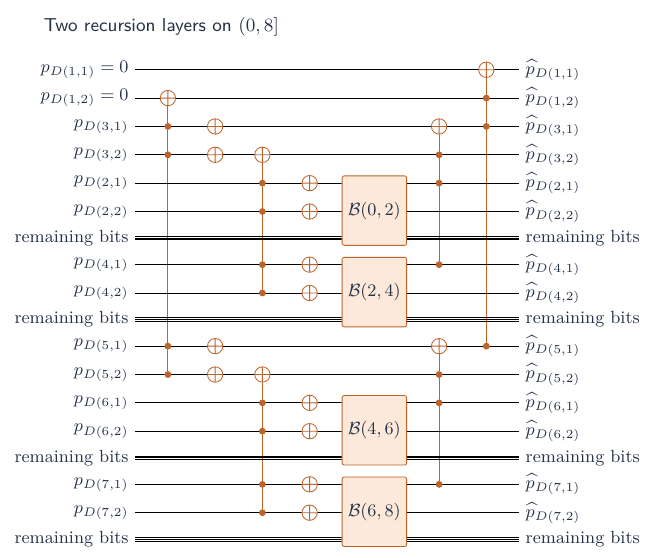}
\caption{Two layers of the interval encoding for $r=8$. The root pair is shown after its initial complements, on the branch where it enters at zero. Each box $\mathcal B(j,i)$ continues the recursion on $(j,i]$, omitting positions chosen at larger intervals. Hats mark the values after encoding.}
\label{fig:interval-encoding}
\end{figure}

We now fix the positions. Use the four bottom controls in each block:
\begin{equation}
\begin{aligned}
 D(i,1)&=(i,b-3),&\qquad D(i,2)&=(i,b-2),\\
 D(i,3)&=(i,b-1),& D(i,4)&=(i,b).
\end{aligned}
\label{eq:block-delegate-indices}
\end{equation}
Apply the following formula separately within each maximal dyadic interval, writing $s$ for its right endpoint. At the interval $(j,i]$, choose the two positions $D(k,t),D(k,t+1)$, where
\begin{equation}
 (k,t)=
 \begin{cases}
 (i,3),&i-j=1,\\[2pt]
 \bigl(j+1+\tfrac{i-j}{2}[i<s],\,1\bigr),&i-j>1.
 \end{cases}
\label{eq:interval-delegate-choice}
\end{equation}
At intervals containing several blocks, an inherited pair lies in the first block. The formula places the new pair in the other half when that inherited pair is present; otherwise it uses the first block. This keeps the two pairs in different children.

For example, Figure~\ref{fig:interval-delegates} starts with the pair $D(1,1),D(1,2)$ at $(0,8]$. The interval $(0,4]$ inherits that pair and chooses $D(3,1),D(3,2)$ in its other half. Its two children then inherit the pairs in $B_1,B_3$ and choose the pairs in $B_2,B_4$. The interval $(4,8]$ starts with the pair in $B_5$, and its children choose the pairs in $B_6,B_7$. For the walk $7\to6\to4\to0$, the intervals $(6,7]$, $(4,6]$, and $(0,4]$ therefore have their encoded bits at $D(7,3)$, $D(6,1)$, and $D(3,1)$, respectively. When $b=64$, these are positions $(7,63)$, $(6,61)$, and $(3,61)$.

\begin{figure}[!ht]
\centering
\makebox[\linewidth][c]{%
\includegraphics[width=1.05\linewidth]{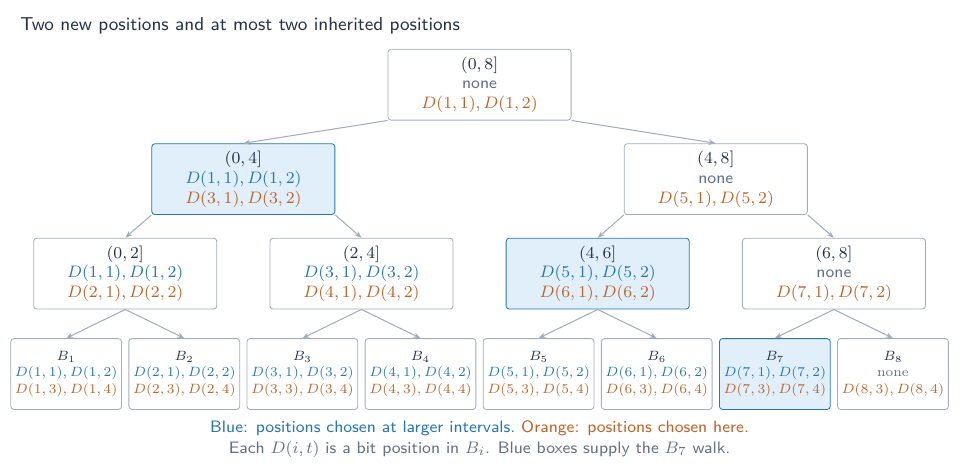}}
\caption{The chosen positions for $r=8$: inherited positions in blue, newly chosen positions in orange. Each $D(i,t)$ is one physical bit index in $B_i$. The highlighted intervals supply the walk $7\to6\to4\to0$.}
\label{fig:interval-delegates}
\end{figure}

Write $\mathcal E:p\mapsto\widehat p$ for the encoding. For an interval $(j,i]$ and its $(k,t)$ from \eqref{eq:interval-delegate-choice}, define
\begin{equation}
\begin{aligned}
 m(j,i)&=p_{D(k,t)}p_{D(k,t+1)}
       \bigl(p_{D(j+1,1)}p_{D(j+1,2)}\bigr)^{[i<s]},\\
 f(j,i)&=\widehat p_{D(k,t)}.
\end{aligned}
\label{eq:shared-interval-flag}
\end{equation}
The parenthesized factor is one when $i=s$. Here $m(j,i)$ is a product of original delegate values, while $f(j,i)$ is read after encoding. These quantities satisfy
\begin{equation}
 \Pi(j,i)=m(j,i)f(j,i).
 \label{eq:shared-interval-product}
\end{equation}

To see why $m$ needs at most four controls, follow the chosen pairs down the tree. An internal interval ending at $s$ inherits no positions; its new pair is in the first block of its left child. The left child inherits that pair, and the right child inherits none. An internal interval ending before $s$ inherits exactly $D(j+1,1),D(j+1,2)$. Its new pair is in the first block of the right child, so each child inherits exactly one pair. Induction therefore gives at most two inherited positions at every interval. Together with its own two chosen positions, these are precisely the at most four controls in $m(j,i)$.

If any of these original controls is zero, both $\Pi(j,i)$ and $m(j,i)$ are zero. Otherwise, the call's own pair enters at zero and all the omitted controls from larger intervals have value one. If any of the four child bits was zero before its complement, the call's second bit stays zero, so its final gate does not act. If all four were one, both children enter with zero pairs, compute the products of their remaining controls, and the final gate multiplies their encoded values. At a leaf this is exactly the single conjunction gate. Thus, when the original delegates are all one, the first encoded bit gives the conjunction of all remaining controls. This proves \eqref{eq:shared-interval-product} by induction. Each child circuit runs regardless of the other child's values, so the identity holds for every interval below the root.

Only the four positions in \eqref{eq:block-delegate-indices} can change, so the encoding preserves $H_i$. There are $O(r)$ internal calls and size-$O(b)$ work at each leaf. Hence $\mathcal E$ has size $O(N)$ and depth $O(\log(N+1))$, where $N=rb$. The first controlled gate of an internal call can borrow its first chosen bit, and the last can borrow a child's second chosen bit; these borrowed bits are restored. A leaf gate can borrow its own second chosen bit.

\paragraph{Computing the endpoint updates.} Use the complete circuit in Figure~\ref{fig:path-echo}, now with shorter walks and the shared interval encoding $\mathcal E$. The position for $s_i$ is again the last bit of $H_i$, namely $(i,6h+2)$. Inside $\mathcal M,\mathcal F$, the circuits $U,V$ now form their products along the forest paths.

For the walk $7\to6\to4\to0$, the three edge controls have the representations in \eqref{eq:shared-interval-flag}--\eqref{eq:shared-interval-product}. Before encoding, $U$ copies their original delegates into $H_7$ and computes
\[
\begin{aligned}
 M(6,7)&=m(6,7),\\
 M(4,7)&=m(6,7)m(4,6),\\
 M(0,7)&=m(6,7)m(4,6)m(0,4).
\end{aligned}
\]
After encoding, $V$ copies the encoded bits and computes $F(6,7),F(4,7),F(0,7)$ by the same products with $f$ in place of $m$. Together with copies of $a_7,Z$, they give
\[
 s_7=a_{7,1}F(6,7)Z_6+a_{7,2}F(4,7)Z_4+
      a_{7,3}F(0,7)Z_0.
\]
In general, along $i=j_0\to j_1\to\cdots\to j_t$, define
\[
 M(j_k,i)=\prod_{\ell=1}^{k}m(j_\ell,j_{\ell-1}),\qquad
 F(j_k,i)=\prod_{\ell=1}^{k}f(j_\ell,j_{\ell-1}).
\]
The word $M(\cdot,i)$ lists these products in walk order, and
\[
 s_i=\sum_{k=1}^{t}
       \left(\binom{d_i}{d_{j_k}}\bmod2\right)
       a_{i,k}F(j_k,i)Z_{j_k}.
\]
Each original edge product $m(j_\ell,j_{\ell-1})$ uses at most four delegates. With only $h$ stops, these running products take $O(h)$ gates and depth. Both $U,V$ compute into $H_i$ from unchanged outside inputs; $U$ leaves the position for $s_i$ untouched. Their inverses restore all temporary copies and products.

The circuits for $\mathcal M$ and $\mathcal F$ use the same zero-test echoes as in Figures~\ref{fig:chain-m-echo} and~\ref{fig:chain-f-echo}. Each $H_i$ is restored before every encoding or decoding, so $\mathcal E$ and its inverse always see the correct control values. The two calls to each stage in Figure~\ref{fig:path-echo} restore $a_i$ and cancel its initial values by \eqref{eq:two-read-cancellation}. At stop $j_k$, the remaining contribution is
\[
 M(j_k,i)F(j_k,i)Z_{j_k}=\Pi(j_k,i)Z_{j_k},
\]
when its Pascal coefficient is odd. If any original control in $H_i$ was zero, every desired term is zero, and the zero-test echoes in both stages make no net change. This completes the outward update and restores all controls and borrowed bits.

For the inward forest, a block's walk starts at a source and follows its successors. Attach $B_i$ to source vertex $i-1$. Its source bit is fixed at $Z_{i-1}$, and each visited vertex specifies an endpoint to change. At a stop $j$, add $\Pi(i-1,j)Z_{i-1}$ when $\binom{d_j}{d_{i-1}}$ is odd. This interval contains $B_i$, so the conditional space is $H_i$. The running products are $M(i-1,j)$ and $F(i-1,j)$, formed along the path from $i-1$ to $j$. Different walks may reach the same endpoint, so we first store their contributions separately. Give each stop a borrowed bit $v_{i,k}$. Here $V$ computes its individual terms in separate positions of $H_i$ and adds each term to its $v_{i,k}$, controlled by $c_i$. We then add the parity of the borrowed bits belonging to each endpoint. For example, the three walks from $B_2,B_3,B_4$ give
\[
\begin{aligned}
 v_{2,2}&\leftarrow v_{2,2}+\Pi(1,4)Z_1,\\
 v_{3,1}&\leftarrow v_{3,1}+\Pi(2,4)Z_2,\\
 v_{4,1}&\leftarrow v_{4,1}+\pi_4Z_3.
\end{aligned}
\]
Add these three bits to $\zeta_4$, run their circuits again to restore them, and add the three bits to $\zeta_4$ again. Their initial values cancel and the three desired terms remain, as Figure~\ref{fig:walk-buffers} shows. Group the bits by the endpoint they change and use CNOT trees for these parity additions.

\begin{figure}[!ht]
\centering
\includegraphics[width=.966\linewidth]{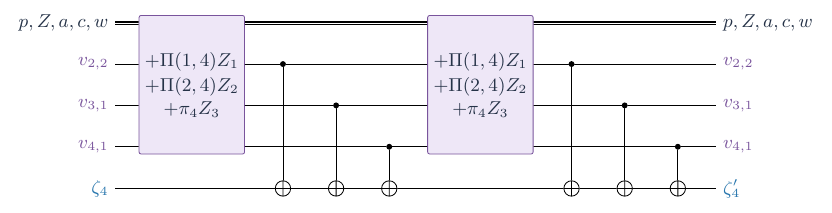}
\caption{Three inward walks contribute to $\zeta_4$. Each box runs the complete circuit with $v_{2,2},v_{3,1},v_{4,1}$ as the bits to change, and restores its other bits. The first three CNOTs add their new parity; the last three add their original parity. All three borrowed bits are restored.}
\label{fig:walk-buffers}
\end{figure}

For $Z\leftarrow Z+DA\zeta$, use $\zeta_j$ as the source and $Z_{i-1}$ in place of an update to $\zeta_i$. Every nonempty path ends at a positive index. For $Q_1$, omit terms whose source has grade zero and decrease both grades by one when selecting the additions.

Each walk has at most $h$ stops. The copies and running products in $U$ use at most $5h+1$ bits of $H_i$. For outward walks, $V$ uses at most $4h+2$ bits, including $s_i$; for inward walks its separate terms use at most $5h+1$ bits. All fit in $6h+2$ positions and take $O(h)$ gates and depth per walk. Copying shared inputs and adding the inward parities take $O(rh)$ gates and logarithmic depth. The encoding is run a constant number of times. We obtain the following bound.

\begin{lemma}[Parallel path update]
\label{lem:batched-path-update}
For either forest and either Pascal factor $Q=I+A$, each of the updates
\[
 (\zeta,Z)\longmapsto(\zeta+AZ,Z),\qquad
 (\zeta,Z)\longmapsto(\zeta,Z+DA\zeta)
\]
has size $O(N+rh)$ and depth $O(\log(N+r+1)+h)$. It uses $6h+2$ control positions per block that the encoding leaves unchanged, one borrowed bit per stop for the outward forest and two for the inward forest, and two more borrowed bits $c_i,w_i$ per walk. All controls and borrowed bits are restored, and the source register is preserved.
\end{lemma}

\subsection{Completing the adder without recursion}
\label{sec:linear-bounds}

Choose a sufficiently large fixed $\kappa\ge64$, let $L=\lceil\log_2(n+1)\rceil$, and take $b$ to be the next power of two at least $\kappa L$. Use the partition of Section~\ref{sec:blocks}, with $r=\lceil n/b\rceil-1$ full blocks after $B_0$. For the finitely many small sizes with $n\le2b$, use the direct ladder. Otherwise $h\le L$, $N=rb\le n$, and
\[
 rh\le\frac{n}{\kappa},\qquad b-4\ge6h+2.
\]
The second inequality leaves enough unchanged control positions for $H_i$ in each block.

The borrowed data also fit inside the existing blocks. Let $q\le rh$ be the total number of stops in all walks. The dirty register $Z$, the words $a_i$, the inward buffers, and the bits $c_i,w_i$ need at most
\[
 (r+1)+2q+2r\le5rh\le\frac{5n}{\kappa}
   \le\frac n2\le n-r-1
\]
bits for sufficiently large $n$. The last quantity is the number of interior data bits. The two forests and their factors run sequentially, reusing the same restored allocation.

We can now finish the three steps of Figure~\ref{fig:block-framework}. Run each local ladder directly: their total size is $O(n)$ and their depth is $O(b)=O(\log n)$. The endpoint update uses sixteen calls to Lemma~\ref{lem:batched-path-update}, each of size $O(n+rh)=O(n)$ and depth $O(\log n)$. Finally, the direct interior correction of Figure~\ref{fig:interior-correction} uses $O(n)$ gates and $O(b)=O(\log n)$ depth. There is only one partition and no recursive call. Lemma~\ref{lem:reduction} gives the same bounds for addition, proving Theorem~\ref{thm:main}.

Every output bit in the second input register must participate in a gate, giving an $\Omega(n)$ size bound for bounded-arity gates. The highest sum bit depends on all $2n$ input bits. Looking backward through one layer, each gate can increase the number of inputs that affect this bit by at most a constant factor. The depth is therefore $\Omega(\log n)$. Our circuit attains both lower bounds, so its size and depth are asymptotically optimal.

\section*{Acknowledgments}
Wei Zi acknowledges support from the National Natural Science Foundation of China (Grant No. 62602563) and the Guangdong Provincial Quantum Science Strategic Initiative (Grant No. GDZX2503001).

\section*{AI statement}
We thank GPT-6 Pro for useful discussions that helped develop and verify the ideas in this paper. In particular, GPT-6 Pro found the forest decomposition that enabled us to achieve asymptotically optimal circuit size. We also used GPT-6.1 Sol (Codex) to assist with writing the manuscript and preparing LaTeX figures. We reviewed, edited, and verified all AI-assisted material and take full responsibility for the content of this paper.

\bibliographystyle{alpha}
\bibliography{ref}

\end{document}